\documentclass[preprint]{elsarticle}

\usepackage{lipsum}
\makeatletter
\def\ps@pprintTitle{% 
 \let\@oddhead\@empty
 \let\@evenhead\@empty
 \def\@oddfoot{}% 
 \let\@evenfoot\@oddfoot}
\makeatother

\usepackage{array,xspace,multirow,hhline,tikz,colortbl,tabularx,booktabs,fixltx2e,amsmath,amssymb,amsfonts,amsthm}
\usepackage{algorithm}
\usepackage{algorithmic}
\usepackage{verbatim,ifthen}
\usepackage{enumitem}
\usepackage{pifont}
\usepackage{ifthen}
\usepackage{calrsfs,mathrsfs}
\usepackage{bbding,pifont}
\usepackage{pgflibraryshapes}

\usepackage{varioref}

\definecolor{light-gray}{gray}{0.9}

	\newcommand{\lemref}[1]{Lemma~\ref{#1}}
	\newcommand{\thmref}[1]{Theorem~\ref{#1}}

	\newcommand{\algref}[1]{Algorithm~\ref{#1}}

	\newcommand{\ie}{i.e.,\xspace}
	\newcommand{\eg}{e.g.,\xspace}

	\newtheorem{lemma}{Lemma}% 
	\newtheorem{theorem}{Theorem}% 
	\newtheorem{corollary}{Corollary}% 

\usepackage{eqparbox}

	\newcommand\eat[1]{}

	\usepackage{enumitem}
	\setenumerate[1]{label=\rm(\it{\roman{*}}\rm),ref=({\it\roman{*}}),leftmargin=*}
	\newlength{\wordlength}
	\newcommand{\wordbox}[3][c]{\settowidth{\wordlength}{#3}\makebox[\wordlength][#1]{#2}}

	\newcommand{\pref}{R\xspace}
	\newcommand{\Pref}[1][]{
		\ifthenelse{\equal{#1}{}}{\mathrel R}{\mathop{R_{#1}}}
	}                                          
	\newcommand{\sPref}[1][]{                  
		\ifthenelse{\equal{#1}{}}{\mathrel P}{\mathop{P_{#1}}}
	}                                          
	\newcommand{\Indiff}[1][]{                 
		\ifthenelse{\equal{#1}{}}{\mathrel I}{\mathop{I_{#1}}}
	}
	\newcommand{\prefset}[1][]{\ifthenelse{\equal{#1}{}}{\mathcal{R}}{\mathcal{R}_{#1}}}

\usepackage{enumitem}
\setenumerate[1]{label=\rm(\it{\roman{*}}\rm),ref=({\it\roman{*}}),leftmargin=*}

\newcommand{\nbh}[1][]{
	\ifthenelse{\equal{#1}{}}{\nu}{\nu(#1)}
}

\newcommand{\cstr}[1][]{
	\ifthenelse{\equal{#1}{}}{\mathscr S}{\cstr(#1)}
}

\newcommand{\choice}[1][]{
	\ifthenelse{\equal{#1}{}}{\mathit{C}}{\choice(#1)}
}

\newcommand{\rsd}[0]{\ensuremath{\mathit{RSD}}}
\newcommand{\sd}[0]{\ensuremath{\mathit{SD}}}

\begin{document}

	\title{The Computational Complexity of \\Random Serial Dictatorship}

	\author[nicta]{Haris Aziz\corref{cor1}}  \ead{haris.aziz@nicta.com.au}		
	\author[tum]{Felix Brandt} \ead{brandtf@in.tum.de}
	\author[tum]{Markus Brill} \ead{brill@in.tum.de}
	\address[nicta]{NICTA and UNSW, 223 Anzac Parade, Sydney, NSW 2033, Australia}
	\address[tum]{Institut f\"ur Informatik, Technische Universit\"at M\"unchen, 85748 Garching, Germany} 
	\cortext[cor1]{Corresponding author}

\begin{abstract}
In social choice settings with linear preferences, \emph{random dictatorship} is known to be the only social decision scheme satisfying strategyproofness and \emph{ex post} efficiency. %\citep{Gibb77a} 
When also allowing indifferences, \emph{random serial dictatorship ($\rsd$)} is a well-known generalization of random dictatorship that retains both properties. $\rsd$ has been particularly successful in the special domain of random assignment where indifferences are unavoidable.
While \emph{executing} $\rsd$ is obviously feasible, we show that \emph{computing} the resulting probabilities is \#P-complete and thus intractable, both in the context of voting and assignment. 
\end{abstract}


\maketitle

\section{Introduction}
	
Social choice theory studies how a group of agents can make collective decisions based on the---possibly conflicting---preferences of its members. In the most general setting, there is a set of abstract \emph{alternatives} over which each agent entertains \emph{preferences}. A \emph{social decision scheme} aggregates these preferences into a probability distribution (or \emph{lottery}) over the alternatives.

Perhaps the most well-known social decision scheme is \emph{random dictatorship}, in which one of the agents is uniformly chosen at random and then picks his most preferred alternative. 
\citet{Gibb77a} has shown that random dictatorship is the only social decision scheme that is strategyproof and \emph{ex post} efficient, \ie it never puts positive probability on Pareto dominated alternatives. Note that random dictatorship is only well-defined when there are no ties in the agents' preferences.
%\footnote{Alternative proofs of this theorem have been provided by \citet{Dugg96a} and \citet{Nand97a}.} 
%A drawback of Gibbard's beautiful result, however, is that it strongly relies on the non-existence of ties in the agents' preferences. 
However, ties are unavoidable in many important domains of social choice such as assignment, matching, and coalition formation since agents are assumed to be indifferent among all outcomes in which their assignment, match, or coalition is the same \citep[\eg][]{SoUn10a}. % AAMAS REF: ABS11b instead of c
% BoLa08a,ElWo09a,ABS11c

In the presence of ties, random dictatorship is typically extended to \emph{random serial dictatorship $(\rsd)$}, where dictators are invoked sequentially and ties between most-preferred alternatives are broken by subsequent dictators.\footnote{$\rsd$ is referred to as \emph{random priority} by \citet{BoMo01a}.} 
$\rsd$ retains the important properties of \emph{ex post} efficiency and strategyproofness and is well-established in the context of random assignment~\citep[see \eg][]{Sven94a,AbSo98a,BoMo01a, CrMo01a}.
%\footnote{The fact that the class of assignment problems corresponds to a restricted domain of social choice setting is illustrated in \secref{sec:domain}.} 

	% Although $\rsd$ has been studied in a number of papers in operations research as well as economics~\citep[see \eg][]{BoMo01a,Budi12a,Mane08a,Mane09a}, it was not clear how easy is it to compute the fractional assignment achieved by $\rsd$. 
	% Even though $\rsd$ is one of the simplest mechanism to state and implement,

In this paper, we focus on two important domains of social choice: (1)~the \emph{voting setting}, where alternatives are candidates and agents' preferences are given by rankings over candidates, which are not subject to restrictions of any kind, and (2)~the aforementioned \emph{assignment setting}, where each alternative corresponds to an assignment of houses to agents and agents' preferences are given by rankings over houses. 
Whereas agents' preferences over alternatives are listed explicitly in the voting setting, this is not the case in the assignment setting. However, preferences over houses can be easily extended to preferences over assignments by assuming that each agent only cares about the house assigned to himself and is indifferent between all assignments in which he is assigned the same house. As a consequence, the assignment setting is a special case of the voting setting. However, due to the different representations, \emph{computational} statements do not carry over from one setting to the other. % more details in full version

We examine the computational complexity of $\rsd$ and show that computing the $\rsd$ lottery is \#P-complete%
\footnote{\#P-completeness is commonly seen as strong evidence that a problem cannot be solved in polynomial time.}  
\emph{both} in the voting setting and in the assignment setting.
As mentioned above, neither of these two results implies the other. % why?
We furthermore present a polynomial-time algorithm to compute the \emph{support} of the $\rsd$ lottery in the voting setting. This is not possible in the assignment setting, because the support of the $\rsd$ lottery might be of exponential size. However, we can decide in polynomial time whether a given alternative (\ie an assignment) is contained in the support or not.
		
\section{Preliminaries}
		
In the general social choice setting, there is a set $N=\{1,\ldots, n\}$ of \emph{agents}, who have preferences over a finite set $A$ of \emph{alternatives}. The preferences of agent~$i \in N$ are represented by a complete and transitive \emph{preference relation}~$R_i\subseteq A\times A$. The interpretation of $(a,b) \in R_i$, usually denoted by $a \mathrel{R_i} b$, is that agent~$i$ values alternative~$a$ at least as much as alternative~$b$. In accordance with conventional notation, we write~$P_i$ for the strict part of~$R_i$, \ie~$a \mathrel{P_i} b$ if~$a \mathrel{R_i} b$ but not~$b\mathrel{R_i} a$, and~$I_i$ for the symmetric part of~$R_i$, \ie~$a \mathrel{I_i} b$ if~$a \mathrel{R_i} b$ and~$b\mathrel{R_i} a$. A \emph{preference profile} $R = (R_1,\dots, R_n)$ is an $n$-tuple containing a preference relation $R_i$ for every agent $i \in N$.
		
A preference relation $R_i$ is \emph{linear} if $a \mathrel{P_i} b$ or~$b\mathrel{P_i} a$ for all distinct alternatives $a,b \in A$. A preference relation is \emph{dichotomous} if $a\mathrel{R_i} b \mathrel{R_i} c$ implies $a\mathrel{I_i} b$ or $b\mathrel{I_i} c$.
		
We let $\Pi^N$ denote the set of all permutations of $N$ and write a permutation $\pi \in \Pi^N$ as $\pi = \pi(1) \ldots \pi(n)$. For $k\leq n$, we furthermore let $\pi|_k$ denote the prefix of $\pi$ of length $k$, \ie $\pi|_k = \pi(1) \ldots \pi(k)$.

If $R_i$ is a preference relation and $B\subseteq A$ a subset of alternatives, then $\max_{R_i}(B)=\{ a\in B \colon a \mathrel{R_i} b \text{ for all }b\in B\}$ is the set of most preferred alternatives from $B$ according to $R_i$. 
Hence, $a \mathrel{I_i} b$ for all $a,b \in \max_{R_i}(B)$ and $a \mathrel{P_i} b$ for all $a \in \max_{R_i}(B), b \in B\setminus \max_{R_i}(B)$.

	In order to define the social decision scheme known as \emph{random serial dictatorship ($\rsd$)}, let us first define its deterministic variant \emph{serial dictatorship $(\sd)$}. For a given preference profile $R$ and a permutation $\pi \in \Pi^N$, $\sd(R,\pi)$ is defined via the following procedure. Agent $\pi(1)$ chooses the set of most preferred alternatives from $A$, $\pi(2)$ chooses his most preferred alternatives from the refined set and so on until all agents have been considered. The resulting set of alternatives is returned.
	Formally, $\sd(R,\pi)$ is inductively defined via $\sd(R,\pi|_0)=A$ and 
	$\sd(R,\pi|_i) = \max_{R_{\pi(i)}} (\sd(R,\pi|_{i-1}))$.

	Throughout this paper, we assume that the preferences of the agents are such that there is no pair $a,b\in A$ with $a\ne b$ and $a\mathrel{I_i} b$ for all $i\in N$.% 
	\footnote{In the assignment setting, this assumption always holds if the agents have linear preferences over houses. 
$\sd$ (and $\rsd$) can be defined without this assumption \citep[\eg][]{ABBH12a}.
	}
	This assumption ensures that the set $\sd(R,\pi)$ is always a singleton. We will usually write $\sd(R,\pi)=a$ instead of $\sd(R,\pi)=\{a\}$.

We are now ready to define $\rsd$. For a given preference profile $R$, $\rsd$ returns $\sd(R,\pi)$, where $\pi$ is chosen uniformly at random from $\Pi^N$. The probability $\rsd(\pref)(a)$ of alternative $a \in A$ is thus proportional to the number of permutations $\pi$ for which $\sd(R,\pi)=a$:
		\[\rsd(\pref)(a)= \frac{1}{n!} \; \left|\left\{\pi\in \Pi^N: \sd(R,\pi)=a\right\}\right| .\]
We refer to the probability $\rsd(R)(a)$ as the \emph{$\rsd$ probability} of alternative $a$ and to the probability distribution $\rsd(R)$ as the \emph{$\rsd$ lottery}.

Our proofs leverage the fact that a certain matrix related to the Pascal triangle has a non-zero determinant.
	
	\begin{lemma}[\citet{Bach02a}]\label{lem:det}
		The $n \times n$ matrix $M=(m_{ij})_{i,j}$ given by $m_{ij}=(i+j-2)!$ has a non-zero determinant. 
		That is,
		\[
		\det 
		\begin{pmatrix}
		 0!& 1! & \cdots & (n-1)!  \\
		 1!& 2! & \cdots & n!  \\
		\vdots  & \vdots  & \ddots & \vdots  \\
		 (n-1)!& n!  & \cdots & (2n-2)!  \\ 
		\end{pmatrix}
		\neq 0\text.
		\]	
	\end{lemma}

\section{Voting Setting}

	A \emph{voting problem} is given by a triple $(N, A, \pref)$, where $N=\{1,\ldots, n\}$ is a set of agents, $A$ is a set of alternatives, and $\pref=(\pref_1,\ldots, \pref_n)$ is a preference profile that contains, for each agent $i$, a preference relation on the set of alternatives. The goal is to choose an alternative that is socially acceptable according to the preferences of the agents. 
	
If each agent has a unique most preferred alternative, the $\rsd$ lottery can be computed in linear time. Therefore, computational aspects of $\rsd$ only become interesting when at least some of the agents express indifferences among their most preferred alternatives. The straightforward approach to compute the $\rsd$ lottery involves the enumeration of permutations. This approach obviously takes exponential time. At first sight, it seems that even finding the support of the $\rsd$ lottery requires the enumeration of all permutations. However, we outline a surprisingly simple algorithm that checks in polynomial time whether a given alternative $a$ is contained in the support (Algorithm~\ref{algo:rsdsupport}).

	\begin{algorithm}[h!]
	  \caption{Is the RSD probability of alternative $a$ positive?}
	  \label{PS}
	\renewcommand{\algorithmicrequire}{\wordbox[l]{\textbf{Input}:}{\textbf{Output}:}} 
	 \renewcommand{\algorithmicensure}{\wordbox[l]{\textbf{Output}:}{\textbf{Output}:}}
	\algsetup{linenodelimiter=\,}
	  \begin{algorithmic}[1] 
		
		\STATE $N'\longleftarrow N$
		\STATE $A'\longleftarrow A$
		\STATE $\pi^*\longleftarrow $ empty list
		\WHILE{$N'\neq \emptyset$}
		\IF{$a\notin \max_{\pref_i}(A')$ for all $i\in N'$}
					\RETURN  ``no''
					\ELSE
					 \STATE Take the smallest $i\in N'$  such that $a\in \max_{\pref_i}(A')$
					\STATE $N'\longleftarrow N'\setminus \{i\}$
\STATE $A'\longleftarrow \max_{\pref_i}(A')$
\STATE Append $i$ to $\pi^*$
						\ENDIF  
						\ENDWHILE
						\RETURN ``yes''					
	 \end{algorithmic}
	\label{algo:rsdsupport}
	\end{algorithm}
	
	The algorithm is based on a greedy approach and maintains a working set of alternatives $A'$ and a working set of agents $N'$, which are initialized as $A$ and $N$, respectively. If no agent in $N'$ has $a$ as a most preferred alternative in $A'$, then the algorithm returns ``no.'' Otherwise let $i\in N'$ be the smallest index such that agent $i$ has $a$ as a most preferred alternative in $A'$.\footnote{One may choose any $i\in N'$ for which $a$ is a most preferred alternative in $A'$ and the algorithm still works. Our choice of $i$ with the smallest index simplifies the proof.} 
	The set $A'$ is refined by deleting all alternatives that are not among the most preferred alternatives in $A'$ according to agent $i$. The process is then repeated until all agents have been considered.	

The following lemma will be essential for showing the correctness of Algorithm~\ref{algo:rsdsupport}.

\begin{lemma} \label{lem:move}
Consider a permutation $\pi \in \Pi^N$ and let $1 \leq k<j \leq n$ be such that $a \in \max_{R_{\pi(j)}} \sd(R,\pi|_k)$. Define another permutation $\pi' \in \Pi^N$ by moving $\pi(j)$ to position $k+1$, \ie $\pi' = \pi(1) \ldots \pi(k) \pi(j) \pi(k+1) \ldots \pi(j-1) \pi(j+1) \ldots \pi(n)$. If $\sd(R,\pi)=a$, then $\sd(R,\pi')=a$.
\end{lemma}

\begin{proof}
	Assume for contradiction that $\sd(R,\pi') \ne a$. Then, there exists $i\le n$ such that $a \in \sd(R,\pi'|_{i-1})$ and $a \notin \sd(R,\pi'|_{i})$. That is, $a\notin \max_{R_{\pi'(i)}}(\sd(R,\pi'|_{i-1}))$.
	Since $\pi$ and $\pi'$ agree on the first $k$ positions and $\sd(R,\pi)=a$ , it follows that $i>k$.

	Now consider the set that agent $\pi'(i)$ faces when it is his turn in the original permutation $\pi$. This set, call it $B$, is identical to the set $\sd(R,\pi'|_{i-1})$ the agent faces in permutation $\pi'$, except that agent $\pi(j)$ might not have refined the set yet. Thus $\sd(R,\pi'|_{i-1}) \subseteq B$. Since $a$ was  not among the most preferred alternatives in $\sd(R,\pi'|_{i-1})$ (according to the preferences of agent $\pi'(i)$), it follows that $a$ is not among the most preferred alternatives in $B$. The latter statement however contradicts the assumption that $\sd(R,\pi)=a$.
\end{proof}

\begin{theorem}\label{thm:voting-support}
	In the voting setting, the support of the $\rsd$ lottery can be computed in polynomial time.
\end{theorem}

\begin{proof}
It is obvious that Algorithm~\ref{algo:rsdsupport} runs in polynomial time. Since the number of alternatives is linear, we can run Algorithm~\ref{algo:rsdsupport} for each alternative and return the set of alternatives for which the algorithm returns ``yes.'' We now show that Algorithm~\ref{algo:rsdsupport} returns ``yes'' if and only if $a$ is in the support of the $\rsd$ lottery. 

If Algorithm~\ref{algo:rsdsupport} returns ``yes,'' it also constructs a corresponding permutation~$\pi^*$ such that for each $j\in N$, $a$ is one of the most preferred alternatives of $\pi^*(j)$ in the working set of alternatives. Hence, $a\in \sd(\pref,\pi^*)$.

For the other direction, assume that $a$ is in the support of the $\rsd$ lottery. Then, there exists a permutation $\pi^a$ such that $\sd(R,\pi^a)=a$. This permutation can be transformed into the permutation $\pi^*$ constructed by Algorithm~\ref{algo:rsdsupport} by only using steps that are covered by \lemref{lem:move}. In particular, we start with permutation $\pi^a$ and move agent $\pi^*(1)$ to position $1$. Then, agent $\pi^*(2)$ is moved to position $2$, and so on until all agents are in position.
Repeated application of \lemref{lem:move} yields that $\sd(R,\pi^*)=a$. Therefore, Algorithm~\ref{algo:rsdsupport} returns ``yes.''
\end{proof}

We now show that the problem of computing the actual $\rsd$ lottery is intractable, even when preferences are severely restricted. 

\begin{theorem}\label{thm:voting-hard}
	In the voting setting, computing the $\rsd$ probability of an alternative is \#P-complete, even for dichotomous preferences.
\end{theorem}

\begin{proof}
	We show that computing the probability $\rsd(R)(a)$ of $\rsd$ choosing alternative $a\in A$ is \#P-complete. Membership of this problem in complexity class \#P is straightforward.  For hardness, we present a polynomial-time Turing reduction from the \#P-complete problem {\sc \#SetCovers}-$k$:

\begin{quote}
Given a set $U$ and a collection $S=\{S_1,\ldots, S_n\}$ of subsets of $U$,
count the number of set covers of $U$ of size $k$. Here, $S'\subseteq S$ is a \emph{set cover} if $\bigcup_{S_i\in S'} S_i=U$ and the \emph{size} of a set cover $S'$ is $|S'|$.
\end{quote}
\#P-completeness of {\sc \#SetCovers}-$k$ follows from the fact that counting set covers of arbitrary size is \#P-complete which in turns follows from the result of \citet{Vadh97a} that even the problem of counting \emph{vertex covers} (a special case of set covers) is \#P-complete.
	
	We are now in a position to outline our Turing reduction.\footnote{Similar reduction techniques are used by \citet{Vali79a}.}
Given an instance $(U,S)$ of {\sc \#SetCovers}-$k$, we construct a preference profile $R^k$ for each $k\in \{1,\ldots, n\}$. The set of alternatives in $R^k$ is $A=U\cup\{a\}$ where $a \notin U$ and the set of agents is $N^k=\{1,\ldots, n+k\}$. The preferences of the agents are specified as follows. Each agent $i\in \{1,\ldots, n\}$ has dichotomous preferences
	\[i\colon (U\setminus S_i)\cup\{a\} \mathrel{P_i} S_i\]
	and each agent $i\in \{n+1,\ldots, n+k\}$ has dichotomous preferences
	\[i\colon U \mathrel{P_i} a \text.\]

We now construct a system of equations in order to show that efficient computability of $\rsd$ probabilities implies efficient computability of the number of set covers of a given size.	For each $k\in \{1,\ldots, n\}$, let $p_k$ denote the number of permutations $\pi \in \Pi^{N^k}$ for which $\sd(R^k,\pi)=a$. Furthermore, let $x_j$ denote the number of set covers of $U$ of size $j$. 
We show that for all $k\in \{1,\ldots, n\}$, 
\begin{equation} \label{eq:p_k}
p_k = \rsd(R^k)(a) \times (n+k)! = \sum_{j=1}^n j! \times k \times (n+k-j-1)! \times x_j \text.	
\end{equation}

The first equality follows directly from the definition of $\rsd$. In order to verify that the third term in equation (\ref{eq:p_k}) equals $p_k$, envision $j+1$ as the earliest position in the permutation at which an agent $\ell \in \{n+1,\ldots, {n+k}\}$ is present. 
	If alternative~$a$ is chosen by $\sd$ for a given permutation, then it must be the case that the $j$ agents preceding $\ell$ in the permutation must have already filtered out all the alternatives in $U$ which means that their corresponding sets must cover $U$. (If the first $j$ agents in a permutation do not filter out all the elements in $U$, then the agent in position $j+1$ ensures that $a$ is \emph{not} chosen since all the alternatives in $U$ are strictly more preferred to $a$ by this agent.) In this case, neither the ordering of the first $j$ agents nor the identity of $\ell \in \{n+1,\ldots, {n+k}\}$ matters. Hence, the number of permutations in which the earliest $\ell \in \{n+1,\ldots, n+k\}$ is at position $j+1$ and $a$ gets selected is $x_j \times j!\times k\times (n+k-j-1)!$. 
Therefore, the total number of permutations of $N^k$ in which $a$ gets selected is $\sum_{j=1}^n j! \times k \times (n+k-j-1)! \times x_j$.

It follows from (\ref{eq:p_k}) that for all $k\in \{1,\ldots, n\}$,
\[\sum_{j=1}^n j!\times (n-j+k-1)! \times  x_j = \rsd(R^k)(a)\times \frac{(n+k)!}{k} \text.\]

We get the following system of equations.	
	\[ 
		\begin{pmatrix}
		(n)!(0)! 	& \cdots &  1!(n-1)! \\
		\vdots   	& \ddots & \vdots  \\
		(n)!(n-1)! 	& \cdots & 1!(n-1+n-1)!  \\ 
		\end{pmatrix}\begin{pmatrix}
		 x_{n}\\
		\vdots  \\
		x_1
		\end{pmatrix}=
		\begin{pmatrix}
		 \rsd(R^1)(a)\times (n+1)!/1\\
		\vdots  \\
	\rsd(R^n)(a)\times (2n)!/n
		\end{pmatrix}\]

If the $\rsd$ probabilities can be computed in polynomial time, so can the $n\times 1$ matrix on the right hand side of the equation. Furthermore, \lemref{lem:det} implies that the $n\times n$ matrix on the left hand side of the equation is non-singular. Hence, the values $x_1,\ldots, x_n$ can be computed in polynomial time via Gaussian elimination. The largest constants in the system of equations are of order $(2n)!$ and can be represented by $O(n\log n)$ bits. Since the variables $x_1,\ldots, x_n$ can be computed in polynomial time, it follows that {\sc \#SetCovers}-$k$ can be solved in polynomial time.
\end{proof}

\begin{corollary}\label{cor:rsdgreaterthank}
The following problem is NP-hard: Given an alternative $a\in A$ and $q\in (0,1)$, is the $\rsd$ probability of $a$ greater than or equal to $q$? 
\end{corollary}
				\begin{proof}
	Let $p=\rsd(R)$ and note that $p(a)=c/n!$ for some integer $c\in [0,n!]$. Therefore $p(a)$ can take $n!$ different values. If there was a polynomial-time algorithm to check whether $p(a)\geq q$, then we can actually compute $p(a)$ via binary search in $\log{n!}<\log{n^n}=n \log{n}$ queries.
				\end{proof}

\section{Assignment Setting}
\label{sec:domain}

We now turn to the setting in which random serial dictatorship is most commonly studied. An \emph{assignment problem} is given by a triple $(N, H, \pref)$, where $N=\{1,\ldots, n\}$ is a set of agents, $H$ is a set of houses with $|H|=n$, and $\pref=(\pref_1,\ldots, \pref_n)$ is a preference profile that contains, for each agent $i$, a \emph{linear} preference relation on the set of houses. The goal is to find an assignment of agents to houses. A \emph{deterministic} assignment is a one-to-one mapping $\sigma: A \rightarrow H$. Apart from deterministic assignments, assignment mechanisms are usually allowed to return randomized assignments, \ie lotteries over the set of deterministic assignments. Every randomized assignment yields a \emph{fractional assignment} that specifies, for every agent $i$ and every house $h$, the probability $p_{ih}$ that house~$h$ is assigned to agent $i$. The fractional assignment can be seen as a compact representation of the randomized assignment.

Observe that the assignment setting corresponds to a special case of the general social choice problem where the set $A$ of alternatives is given by the set of deterministic assignments and agents' preferences over $A$ are obtained by extending their preferences over $H$ in such a way that each agent is indifferent between all deterministic assignments in which he is assigned the same house.\footnote{While agents' preferences over $H$ are linear, their extended preferences over $A$ are not.}

Executing $\rsd$ in the assignment setting has a particularly natural interpretation: choose an ordering of the agents uniformly at random and let every agent select his most preferred among all remaining houses. 

In this section, we examine the computational complexity of $\rsd$ in the assignment setting. 
Listing the probability of each deterministic assignment explicitly is prohibitive because the number of deterministic assignments is exponential in the size of the problem instance. However, it can be checked in polynomial time whether a distinguished deterministic assignment is realized with a positive probability.

			\begin{theorem}
It can be checked in polynomial time whether some deterministic assignment is in the support of the $\rsd$ lottery. 
			\end{theorem}

	We omit the straightforward proof, which invokes a simplified version of \algref{algo:rsdsupport}.

We now show that computing the fractional assignment is intractable. We first present a useful lemma. Fix a house $h \in H$ and let $s_j$ denote the number of sets $N' \subseteq N$ with $|N'|=j$ such that there exists a deterministic assignment in which each agent in $N'$ gets a house that he prefers to $h$.

\begin{lemma}\label{lemma:alpha}
	Computing $s_j$ is \#P-complete.
\end{lemma}

\begin{proof}
\citet{CPV95a} proved that the following problem is \#P-complete: 
	\begin{quote}
Given an undirected and unweighted bipartite graph $G=(S\cup T,E)$ with $E\subseteq S\times T$, compute $x$, the number of subsets $B\subseteq S$ such that $(B\cup T,E)$ contains a perfect matching. 
	\end{quote}
We propose a polynomial-time Turing reduction from computing $x$ to computing~$s_j$. 
	Consider an assignment problem $(N, H,\pref)$ in which $N=S$, $H = T \cup \{h\}$, and for any $i\in S$ and $h'\in T$, $h' \mathrel{P_{i}} h$ if and only if $\{i,h'\}\in E$. Then, $x=\sum_{j=1}^{|S|}s_j$. Therefore, if there exists a polynomial-time algorithm to compute $s_j$, then there exists a polynomial-time algorithm to compute $x$.
\end{proof}

\begin{theorem}\label{thm:assign-hard}
	Computing the probability with which an agent gets a certain house under $\rsd$ is \#P-complete.
\end{theorem}

\begin{proof}
	For an instance $G=(N,H,\pref)$ of the assignment problem, we denote by $(p_{ih}(G))_{i \in N, h \in H}$ the fractional assignment resulting from $\rsd$.
	
Consider an instance $G_0=(N,H,\pref)$ of the assignment problem and fix an agent $i \in N$. Let $h$ be the most preferred house of agent $i$. We show that computing the probability $p_{ih}$ is \#P-complete. 

	By $G_k$ we will denote a modification of $G_0$ in which $k$ additional agents~$N_k$ and houses have been added. The $k$ new agents have the same preferences over~$H$ as agent $i$. The $k$ additional houses are each less preferred than houses in $H$ by all agents in $N\cup N_k$ in any arbitrary order. 

	We now present a polynomial-time reduction from computing $s_j$ to computing the probability $p_{ih}$ in $G_0$. In problem $G_k$, for $i$ to get house $h$, it is clear that $i$ should be earlier in the permutation than all agents in $N_k$. Furthermore, agents in $N\setminus \{i\}$ which are earlier than $i$, should all be able to get a house that they prefer to house $h$. Based on this insight, we obtain the following equations.
	\begin{align*}
		p_{ih}(G_0) &= \frac{1}{n!} \sum_{j=0}^{n-1}s_j 		  \times j! \times (n-j-1)!\\
		p_{ih}(G_1) &= \frac{1}{(n+1)!} \sum_{j=0}^{n-1}s_j 	  \times j! \times (n-j-1+1)!\\
					&\vdots\\                                        
		p_{ih}(G_{n-1}) &= \frac{1}{(2n-1)!} \sum_{j=0}^{n-1}s_j \times j! \times (n-j-1+n-1)!
		\end{align*}
		The system of equations can be written as

		\begin{equation*}\label{eqn:determinant}
		M 
		\begin{pmatrix}
		 s_{n-1}\\
		s_{n-2}\\
		\vdots  \\
		s_{0}\\
		 \end{pmatrix}
		 =
		\begin{pmatrix}
		 p_{ih}(G_0)\times n!\\
		p_{ih}(G_1)\times (n+1)!\\
		\vdots  \\
		p_{ih}(G_{n-1})\times (2n-1)!\\
		 \end{pmatrix} \text,
		\end{equation*}		
where
\begin{equation*}
	M = \begin{pmatrix}
	 (n-1)!0!& (n-2)!1! & \cdots & 0!(n-1)!  \\
	 (n-1)!1!& (n-2)!2!  & \cdots & 0!n!  \\
	\vdots  & \vdots  & \ddots & \vdots  \\
	 (n-1)!(n-1)!& (n-2)!n!  & \cdots & 0!(2n-2)!  \\ 
	\end{pmatrix} \text. 
\end{equation*}

\lemref{lem:det} yields that the $n\times n$ matrix $M$ is non-singular. The remaining arguments are similar to those in the proof of Theorem~\ref{thm:voting-hard}.
	\end{proof}

As in the voting setting, we get an NP-hardness result as an immediate corollary. The proof mirrors that of Corollary~\ref{cor:rsdgreaterthank}.

\begin{corollary}
	The following problem is NP-hard: Given an agent $i \in N$, a house $h \in H$, and $q\in (0,1)$, is the probability $p_{ih}$ greater or equal to $q$?
\end{corollary}

	\thmref{thm:assign-hard} also implies that computing the fractional assignment achieved by the \emph{employment by lotto} method~\citep{ACL99a,KlKl06a} is \#P-complete. This is due to the fact that employment by lotto is an extension of $\rsd$ to two-sided matching in which the set of feasible matchings is a subset of stable matchings. Employment by lotto reduces to $\rsd$ in the assignment model if one side is indifferent among all matchings. Another corollary is that computing the fractional assignment achieved by the \textit{draft} mechanism~\citep{Budi12a} is \#P-complete since the draft mechanism is equivalent to $\rsd$ for the case in which each agent can get a maximum of one house.

\section*{Acknowledgments}

This material is based upon work supported by
the Australian Government's
Department of Broadband, Communications and the Digital
Economy, the Australian Research Council, the Asian
Office of Aerospace Research and Development through
grant AOARD-124056, and the Deutsche Forschungsgemeinschaft under grants {BR~2312/7-1} and {BR~2312/10-1}. 

% \section*{References}

\bibliographystyle{plainnat}

\end{document}